\documentclass[12pt]{article}
\usepackage[a4paper,left=2cm,right=2cm,top=2.5cm,bottom=2.5cm]{geometry}
\usepackage[utf8]{inputenc}
\usepackage{nopageno}
\usepackage{graphicx,float,caption,subcaption}
\usepackage{minted}
\usepackage[dvipsnames]{xcolor}
\usepackage{hyperref}
\hypersetup{
    colorlinks=true,
    linkcolor=blue,
    filecolor=magenta,      
    urlcolor=blue,
    citecolor=PineGreen,
    pdftitle={Bootstrapping Thermality},
    pdfpagemode=FullScreen,
    }

\usepackage{amsmath,amssymb,amsthm,bm,graphicx,float,array,multirow,multicol,rotfloat,hyperref,cleveref,enumerate,geometry,mathdots,adjustbox,booktabs,mathtools,tikz,tikz-cd,pdflscape,csquotes,lscape,rotating,empheq,mathrsfs}
\usetikzlibrary{calc, arrows.meta, shapes.geometric, patterns, positioning, matrix, backgrounds, decorations.pathreplacing, shadows}
\usepackage{shadethm}
\usepackage[para]{threeparttable}
\usepackage[all]{xy}
\usepackage[normalem]{ulem}
\usepackage[numbers,sort&compress]{natbib}
\usepackage{physics}
\usepackage{algorithm}
\usepackage{algorithmic}
\usepackage{authblk}
\newtheorem{thm}{Theorem}

\newtheorem{prop}[thm]{Proposition}

\title{Approximate quantum error correction, eigenstate thermalization and the chaos bound}
\author[1]{Shozab Qasim}
\author[2,3]{Jason Pollack}
\affil[1]{Department of Physics, Dahlem Center for Complex Quantum Systems, Freie Universität Berlin}
\affil[2]{Department of Electrical Engineering and Computer Science, Syracuse University}
\affil[3]{Institute for Quantum \& Information Sciences, Syracuse University}
\date{}

\begin{document}

\maketitle
\vspace{-4em}
\begin{abstract}
Quantum error correction, thermalization, and quantum chaos are fundamental aspects of quantum many-body physics that have each developed largely independently, despite their deep conceptual overlap. In this work, we establish a precise link between all three in systems that satisfy the eigenstate thermalization hypothesis (ETH) and exhibit a well-defined hierarchy of time scales between dissipation and scrambling. Building on the ETH matrix ansatz and the structure of the out-of-time-order correlator (OTOC), we show that the chaos bound directly constrains the error of an approximate quantum error-correcting code. This establishes a quantitative relation between information scrambling, thermalization, and correctability. Furthermore, we derive bounds on dynamical fluctuations around the infinite-time average and on fluctuation–dissipation relations, expressed in terms of both the code error and the Lyapunov exponent. Our results reveal how the limits of quantum chaos constrain information preservation in thermalizing quantum systems.
\end{abstract}
\tableofcontents 

\section{Introduction}
The theory of quantum information has profoundly enriched our understanding of quantum many-body systems. One striking example is the problem of thermalization, and more broadly, the emergence of statistical mechanics from unitary quantum dynamics \cite{Gogolin_2016}. Thermalization in isolated quantum systems poses a formidable challenge, as it requires access to the full many-body spectrum rather than only ground-state properties (but see \cite{Garrison_2018,Qi_2019} for cases in which these might be equivalent). A major conceptual advance in addressing this problem was the formulation of the eigenstate thermalization hypothesis (ETH) \cite{Srednicki:1994mfb,Deutsch:1991msp,Jensen:1984gu, D_Alessio_2016}. ETH posits a universal structure for the matrix elements of low-complexity operators in the energy eigenbasis of a chaotic many-body Hamiltonian. It provides a microscopic mechanism for thermalization by asserting that the expectation values of few-body observables in individual eigenstates coincide with their thermal averages. In systems that satisfy ETH, eigenstates, and their eigenvalues, exhibit more structure than those of random matrix Hamiltonians.

Remarkably, the ETH matrix ansatz has also been connected to the framework of approximate quantum error correction (AQEC) \cite{Schumacher:2001xmz}. Originally developed to safeguard quantum information against noise, quantum error correction has since found far-reaching applications in many-body physics, topological phases, quantum chaos, and holography \cite{Kitaev:1997wr, Brand_o_2019, Hayden:2007cs, Harlow:2016vwg, Almheiri:2014lwa}. Brandao et al. \cite{Brand_o_2019} demonstrated that the ETH ansatz satisfies the approximate Knill–Laflamme conditions \cite{B_ny_2010}, thereby establishing a formal equivalence between thermalization and the robustness of logical information encoded in chaotic eigenstates against local noise. Within this interpretation, the error of the code quantifies the degree to which orthogonality within the code subspace is preserved under the action of local operators. In many-body physics terms, this error captures the deviation of real system eigenstates from the idealized random vectors of random matrix theory, which is an imprint of the additional structure in ETH-obeying chaotic systems.

The ETH is deeply intertwined with quantum chaos, as it describes the statistical properties of matrix elements in chaotic many-body systems. Recent works have focused on quantifying dynamical operator growth processes, notably through the out-of-time-order  correlator (OTOC), which measures information scrambling. Maldacena, Shenker, and Stanford \cite{Maldacena_2016} argued that in thermal quantum systems the rate of exponential growth of the OTOC, the Lyapunov exponent $\lambda$, is universally bounded as $\lambda \leq \frac{2\pi}{\beta}$.
This bound on chaos is saturated in the Sachdev–Ye–Kitaev (SYK) model \cite{Sachdev_1993,kitaev2015talks, Maldacena:2016hyu} and in large-$N$ conformal field theories with gravitational duals \cite{Maldacena:1997re, Witten:1998qj}, linking chaotic dynamics to black hole information scrambling \cite{Hayden:2007cs, Sekino:2008he}. Subsequently, Murthy and Srednicki \cite{Murthy_2019} demonstrated that the same bound can be derived directly from the ETH framework, showing that thermalization and chaos are two aspects of a single microscopic structure.

In this work, we establish a concrete bridge between these perspectives. Starting from the ETH formulation and its equivalence to AQEC, we show that the code error associated with the approximate Knill–Laflamme conditions is directly controlled by the Lyapunov exponent and the chaos bound:
\begin{equation}
\varepsilon_{\mathrm{code}} \lesssim 2^{d + 2k} \cdot \exp\left( -\frac{S}{4} -\frac{\pi}{2\lambda} |\omega| \right).
\end{equation}
We further express dynamical fluctuations, static fluctuations, and fluctuation–dissipation relations in terms of these quantities, revealing how the structure of quantum chaos constrains the fidelity of information preservation in chaotic eigenstates.
We emphasize that, like \cite{Murthy_2019}, our analysis holds within a well-defined hierarchy of timescales between the scrambling time and the dissipation time for systems satisfying ETH for which the approximate quantum error correcting code paradigm is expressed for a local error model. In general we expect a large ratio between these timescales to be associated with a small parameter, large-$N$ limit, or other fine-tuning of the dynamics.

\subsection{Relation to previous results}
A precursor to our result, though without connection to the chaos bound, appeared in the work of Bao et al. \cite{Bao_2019}. Our analysis extends their framework in two key ways.
First, we refine the computation of the error term in the approximate quantum error-correcting code by explicitly incorporating the dependence on the spectral function. This additional structure enables us to relate the code error directly to the Lyapunov exponent and the chaos bound, thereby linking quantitative measures of quantum chaos to information-theoretic quantities.
Second, we address an open problem highlighted in \cite{Bao_2019}, namely, relating quantitative signatures of chaos, such as the decay of OTOCs and operator growth, to the formalism of quantum error correction.

Our work also complements and extends the results of Murthy and Srednicki \cite{Murthy_2019}, who derived the chaos bound from the ETH ansatz. By reinterpreting their derivation through the lens of approximate quantum error correction, we provide an information-theoretic perspective on the origin of the chaos bound within ETH, thereby unifying the ETH, quantum chaos, and AQEC frameworks under a common formal structure.

\subsection{Organization}
The rest of the paper is organized as follows. In section \ref{gensetup} we provide the necessary ingredients to understand our results. In section \ref{results} we stir these ingredients to derive the relation between the code error and the Lyapunov exponent. We conclude in \ref{discussion} by discussing implications for quantum information, thermalization and chaos in many-body systems, and quantum gravity.

\section{Preliminaries}
\label{gensetup}
We begin with the ETH and fluctuation dissipation relations which are reviewed in \cite{D_Alessio_2016}.

\subsection{Eigenstate Thermalization Hypothesis (ETH) and fluctuations}

The ETH \cite{Srednicki:1994mfb} is formulated as a matrix element ansatz for operators in the energy eigenbasis of a non-integrable or chaotic Hamiltonian, 
\begin{equation}
    \bra{E_m} \mathcal{O}\ket{E_n} = O_{\mathcal{O}}(\overline{E})\delta_{mn} + e^{-S(\overline{E})/2}f_\mathcal{O}(\overline{E},\omega)R_{mn},
\end{equation}
where $\overline{E}=\frac{E_m+E_n}{2}$ and $\omega=\frac{E_n-E_m}{2}$. In the expression, $S(\overline{E})$ is the thermodynamic entropy at energy $\overline{E}$, the functions $O_{\mathcal{O}}(\overline{E})$, $f_{\mathcal{O}}(\overline{E}, \omega)$ are smooth and dependent on the operator $\mathcal{O}$, and the $R_{mn}$ are drawn from a Gaussian distribution with zero mean and unit variance. In the original statement of ETH, the operator $\mathcal{O}$ was taken to be a ``few-body" local observable, in the sense that $\mathcal{O}$ is a local $n$-body observable in a system of $N$ particles with $N \gg n$. Roughly speaking, these are the observables that can be studied in the lab. We expect ETH to hold in narrow energy windows in the middle of the spectrum. 

The spectral function $f_\mathcal{O}(E,\omega)$ encodes static and dynamic operator fluctuations, which we will now discuss. Our presentation follows Section 5.2 of \cite{Pieter}. We will begin with the dynamic fluctuations. ETH implies an ensemble equivalence between the Gibbs ensemble and the diagonal ensemble for long-time averages of subsystem observables. This allows for the possibility that a system fails to exactly equilibrate but rather continues to exhibit large fluctuations around the time-averaged thermal value. According to ETH, the thermal value is the constant equilibrium value. The long time average of the fluctuations captures the amount by which the expectation value of $\mathcal{O}(t)$ typically deviates from the thermal prediction: 
\begin{equation}
    \delta O_t^2 = \lim_{T \xrightarrow{} \infty} \frac{1}{T}\int_0^T dt [\langle O(t) \rangle - \langle O \rangle_\beta]^2 = \lim_{T \xrightarrow{} \infty} \frac{1}{T}\int_0^T dt \Big| \sum_{n,m \neq n} c_n^* c_m O_{mn}e^{i(E_n-E_m)t} \Big|^2. 
\end{equation}
Assuming $E_m-E_n = E_{m'}-E_{n'}$ only if $m=m'$ and $n=n'$, one finds that 
\begin{equation}
    \delta O_t^2 = \sum_{n,m} |c_n|^2 |c_m|^2 |O_{mn}|^2 \leq \max |O_{mn}|^2 = e^{-S(E)}|f_O(\overline{E}, \omega)|^2,
\end{equation}
so the fluctuations around the thermal value are exponentially suppressed.

We will now consider the static fluctuations. Suppose one can prepare the system in an eigenstate of the Hamiltonian. Measuring the observable will lead to different measurement outcomes and operator expectation values correspond to measurement outcome averages. These measurement outcome fluctuations are given by  
\begin{equation}
    \delta O_n^2 = \bra{n}O^2 \ket{n} - \bra{n} O \ket{n} = \sum_{m \neq n} |O_{mn}|^2.
\end{equation}

The appearance of the off-diagonal elements is natural, as they encode the fluctuations away from the diagonal ensemble. After plugging in the ETH ansatz, one finds that 
\begin{equation}
    \delta O_n^2 = \sum_{m \neq n} e^{-S(E_n + \omega/2)} |f_O(E_n + \omega/2, \omega)|^2 R_{mn}^2.
\end{equation}
Where $\omega = (E_m-E_n)/2$ such that $E=(E_m+E_n)/2 = E_n + \omega/2$. For a sufficiently large density of states, one can replace $R_{mn}^2$ with its average - following a replacement of summation with an integral, introducing a corresponding density of states and then performing a Taylor expansion and defining $\beta=S'(E_n)$, one obtains 
\begin{equation}
    \delta O_n^2 = \int_{-\infty}^{+\infty} d \omega e^{\beta \omega /2} \left( |f_O(E_n, \omega)|^2 + \frac{\omega}{2}\frac{\partial}{\partial E}|f_O(E_n, \omega)|^2 \right).
\end{equation}
The Taylor expansion is  motivated from the exponential decay of the spectral function for large $\omega$. The fluctuations of the observable are a smooth function of their energy, and can be written as a weighted integral over all possible fluctuation frequencies. Small fluctuations in the steady state can be approximated by fluctuations by fluctuations of eigenstates with the same energy. This goes by the name of ensemble equivalence as thermal fluctuations associated with the canonical ensemble correspond to the eigenstate fluctuations associated with the microcanoncial ensemble.

Another important result in statistical mechanics are the fluctuation-dissipation relations, which state that the way that a system dissipates energy when perturbed is completely determined by how it fluctuates in equilibrium. In quantum chaotic systems, the relations are directly captured by the ETH structure of matrix elements. The spectral function determines both equilibrium and dynamical fluctuations. The spectral fluctuations also appears in linear response, underlies the fluctuation-dissipation relations and determines nontrivial operator dynamics. We will provide a high-level derivation of these quantities; details can be found in \cite{D_Alessio_2016}.

We start with the anticommutator 
\begin{equation}
    F(t) = \frac{1}{2} \langle \{ O(t), O\} \rangle_\beta - \langle O \rangle_\beta^2,
\end{equation}
which is a quantification of the extent to which the system is in thermal equilibrium, and the commutator 
\begin{equation}
    \rho(t) = \langle[O(t),O] \rangle_\beta,
\end{equation}
which determines the dissipation in the system. The Fourier transforms\footnote{Here we use the Fourier transform convention $\widetilde{F}(\omega) \equiv \frac{1}{2\pi} \int d \omega e^{i \omega t}F(t)$. } $\widetilde{F}(\omega)$ and $\widetilde{\rho}(\omega)$ thus characterize, respectively, the thermal fluctuations at a given frequency and the system's response to an external perturbation at a given frequency. The former quantity $\widetilde{F}(\omega)$ is the quantum equivalent of the power spectrum. The ETH derivation of $\widetilde{F}(\omega)$ and $\widetilde{\rho}(\omega)$, assuming a sufficiently dense spectrum, gives
\begin{equation}
    \widetilde{F}(\omega) = 4 \pi\cosh(\frac{\beta \omega}{2})|f_O(E_\beta, \omega)|^2, \quad \widetilde{\rho}(\omega) = 2\pi \sinh(\frac{\beta \omega}{2})|f_O(E_\beta, \omega)|^2.
\end{equation}
The fluctuation-dissipation relation is then 
\begin{equation}
    \widetilde{F}(\omega) = 2 \coth(\frac{\beta \omega}{2}) \widetilde{\rho}(\omega).
\end{equation}

\subsection{Bound on chaos from ETH}
The ETH operator matrix structure can also be used to derive the bound on the growth rate of the out-of-time-order four-point correlator in chaotic many-body quantum systems as shown by Murthy and Srednicki \cite{Murthy_2019}. The thermally regulated four point out-of-time ordered correlator is given by 
\begin{equation}
    F_{\text{OTO}}(t):= \Tr[\rho^{1/4}A(t)\rho^{1/4}A(0)\rho^{1/4}A(t)\rho^{1/4}A(0)],
\end{equation}
where $A(t)=e^{iHt}Ae^{-iHt}$ is a local operator in the Heisenberg picture, $H$ is the Hamiltonian, $\rho:=e^{-\beta H}/Z$ is a thermal density operator at inverse temperature $\beta$, and $Z:=\Tr e^{-\beta H}$ is the partition function. In the following discussion we set $\hbar=k_B=1$. 

The derivation of \cite{Murthy_2019} is for systems with a scrambling time $t_s$ that is large compared to the dissipation time that characterizes the decay rate of two-point correlators. In this case, for intermediate times $t_d \ll t \ll t_s$, one expects
\begin{equation}
    F_{\text{OTO}}(t) \propto 1-e^{\lambda (t-t_s)},
\end{equation}
where $\lambda$ is the growth rate which is analogous to the Lyapunov growth rate of the deviation of nearby classical trajectories in chaotic systems. The time scale hierarchy of $t_s \gg t_d$ arises in systems with a small parameter $\epsilon$ that determines the scrambling time as $t_s \sim \lambda^{-1}\ln(1/\epsilon)$. Systems of this type with such a small parameter are the SYK model \cite{kitaev2015talks} and large-$N$ conformal field theories with gravitational duals \cite{Maldacena:1997re, Witten:1998qj}. 

The chaos bound conjecture by Maldacena et al \cite{Maldacena_2016} was a bound on the growth rate of the Lyapunov exponent as
\begin{equation}
\lambda \leq 2 \pi/ \beta.
\end{equation}
Murthy and Srednicki then showed \cite{Murthy_2019} that the bound on the growth rate of the out of time ordered correlator could be derived directly from ETH. Their methodology involved using properties of the ETH matrix elements to place a bound on the Fourier transform of the OTOC at high frequencies. This bound can then be used to infer bounds on the OTOC at intermediate times with additional dependencies on its precise functional form. 

We summarize their argument here. The precise claim for the derivation is that the exponential growth rate, if nonzero, is bounded by the chaos bound in systems that obey the ETH. The assumptions for this derivation can be weakened if the envelope function $f(E,\omega)$ loses its smoothness at low frequencies. The derivation works for observables $A$ for which $\mathcal{A}(E)=0$, either due to a symmetry, or by subtracting $\Tr \rho A$ from $A$. This justifies a zero microcanonical average $\mathcal{A}(E)=0$ from the start. Note that $E$ is treated as an extensive quantity and $\omega$ as an intensive quantity. The starting point is to consider a thermally regulated two-point correlator for such an observable $A$ at inverse temperature $\beta$, 
\begin{equation}
    F_2(t):= \Tr[\rho^{1/2}A(t)\rho^{1/2}A(0)].
\end{equation}
which can only be finite if $f(E,\omega)$ fall at large $|\omega|$ at least as fast as $f(E,\omega) \sim \exp(-\beta |\omega|/4)$. The next step in the derivation involves the general four-point correlator for a single observable $A$ at inverse temperature $\beta$. Ref. \cite{Murthy_2019} specialize to the four-point OTOC for times large compared to the dissipation time, which would be comparable or larger than $\beta$, and show it is given by 
\begin{align}
\widetilde{F}_{\text{OTOC}}(\omega) &:= \int_{-\infty}^{+\infty} \frac{dt}{2\pi}e^{-i \omega t}F_{\text{OTOC}}(t) 
\end{align}
They then consider the large-frequency behavior of $f$ and a function that is dependent on the statistical average of  $\overline{R_{ij}R_{jk}R_{kl}R_{li}}$. We note here that the derivation requires a restriction on the family of functions that includes the OTOC. Depending upon the family of functions that is used, one can obtain a stronger bound than the original Maldacena bound. 

The final statement of their results is that the exponential growth rate, if nonzero, is bounded by the Maldacena bound in systems that obey ETH. Given a restriction on the behavior of the spectral function. This restriction, which will be important for constraining the properties of the error correcting code below, is given by 
\begin{equation}
\label{specfuncrest}
    f(E,\omega) \sim \exp(-\frac{\beta}{4} |\omega|) \lesssim \exp\left( -\frac{\pi}{2\lambda} |\omega| \right).
\end{equation}

\subsection{Approximate Quantum Error-Correcting Codes (AQECC)}
The final prerequisite to understand our results is the notion of an approximate quantum error-correcting code (AQECC), for which the relation to ETH was developed by Brandao et al \cite{Brand_o_2019}. 

Our presentation here closely follows the structure of \cite{Brand_o_2019}, Appendix A. Consider $N$ qubits arranged in a line and assume that errors are local. A subspace $\mathcal{C}$ of a $2^N$-dimensional vector space is an $[[N,k,d,\epsilon]]$ AQECC if $\dim(\mathcal{C})=2^k$ and for every quantum channel (i.e., completely positive trace-preserving linear map) $\Lambda$ acting on at most $d$ consecutive qubits, we have \begin{equation}
\min_{\ket{\psi} \in \mathcal{C}^{\otimes2}} \max_{\mathcal{D}}\bra{\psi} (\mathcal{D} \circ \mathcal{N} \otimes I) (\ket{\psi}\bra{\psi})\ket{\psi} \geq 1 - \epsilon,
\end{equation}
where the maximum is taken to be over decoding channels $\mathcal{D}$ and the minimum is taken over pure entangled joint states of $\mathcal{C}$ and an equally-sized reference system. The condition thus states that if one can correct the effect of local noise on at most $d$ qubits up to error $\epsilon$, for a particular $\mathcal{N}$, then the code is $\epsilon$-correctable under $\mathcal{N}$.

Knill and Laflamme \cite{Knill_2000} gave a set of necessary and sufficient conditions for a code to be able to correct a noisy channel $\mathcal{N}(X)= \sum_k E_k X E_k^\dagger$. A set of similar conditions were found by Beny and Oreshkov for the approximate case \cite{B_ny_2010}. Given two channels $\mathcal{N}$ and $\mathcal{M}$, let $d(\mathcal{N}, \mathcal{M}) = \sqrt{1-F(\mathcal{N},\mathcal{M})}$ be the Bures metric, where the fidelity of the two channels $F(\mathcal{N},\mathcal{M})$ is defined using the Uhlmann fidelity: 
\begin{equation}
    F(\mathcal{N}, \mathcal{M}) := \max_{\ket{\psi}} F(I \otimes \mathcal{N}(\ket{\psi}\bra{\psi}), I \otimes \mathcal{M}(\ket{\psi}\bra{\psi})),
\end{equation}
where the maximization is over all bipartite states $\ket{\psi}$ of the input to the channel together with vector space that is isomorphic to it. We then have 
\begin{prop} (\cite{Brand_o_2019}, Proposition 4)
A code defined by the projector $P$ is said to be $\epsilon$-correctable under a noise channel $\Lambda$ iff 
\begin{equation}
    P E_i^\dagger E_j P = \lambda_{ij}P + P B_{ij}P,
\end{equation}
where $\lambda_{ij}$ are the components of a density operator, and $d(\mathcal{N} + \mathcal{B}, \mathcal{N}) \leq \epsilon$, where $\mathcal{N}(\rho)=\sum_{i,j} \lambda_{ij} \Tr(\rho)\ket{i}\bra{j}$ and $(\mathcal{N}+\mathcal{B}(\rho))=\mathcal{N}(\rho) + \sum_{i,j}\Tr(\rho B_{ij})\ket{i}\bra{j}$.
\end{prop}
In this proposition, the projector $P$ is the projector onto the support of states in $\mathbb{C}^{2^N}$ such that for all $i,j$ and any $d-$local operator $E$,
\begin{equation}
    \bra{\psi_i}E\ket{\psi_j} = C_E \delta_{ij} + \epsilon_{ij},
\end{equation}
with $C_E$ a constant that depends only on $E$. Then $\mathcal{C}:= \text{span}\{ \ket{\psi_1}, \dots, \ket{\psi_{2^k}}\}$ forms a $[[N,k,d,2^{d+2k} \max_{i,j} \epsilon_{ij}^{1/2}]]$ AQECC. This is Corollary 5 of \cite{Brand_o_2019}.

\section{Results}
\label{results}

Having assembled the necessary preliminaries, we now give our results. We begin by expressing the AQECC code parameters in terms of the chaos bound, and vice versa. Then we express fluctuations and the fluctuation-dissipation theorems using these quantities.

\subsection{Quantum error correction and the chaos bound}

\begin{thm}
\label{mainthm}
[Chaos--AQEC Tradeoff from ETH]
Let \( \mathcal{C} \subset \mathcal{H} \) be a code subspace of dimension \( 2^k \), spanned by energy eigenstates \( \{ \ket{E_i} \} \) in a microcanonical shell of a chaotic Hamiltonian satisfying the ETH ansatz:
\begin{equation}
    \bra{E_i} A \ket{E_j} = A(E) \delta_{ij} + e^{-S(E)/2} f(E, \omega) R_{ij}, \quad \omega = E_i - E_j,
\end{equation}
for a local observable \( A \), where \( f(E, \omega) \) is a smooth envelope function of frequency, \( R_{ij} \) are drawn from a Gaussian distribution with with zero mean and unit variance, \( S(E) \) is the microcanonical entropy at energy \( E \), and the system saturates the chaos bound \( \lambda = \frac{2\pi}{\beta} \) with inverse temperature \( \beta \). For time scales $t_d<<t<<t_s$ i.e. those that lie between the dissipation time and scrambling time, the AQEC fidelity error \( \varepsilon_{\mathrm{code}} \) under any \( d \)-local error operator \( A \),  satisfies the lower bound:
\begin{equation}
\label{resi}
    \varepsilon_{\mathrm{code}} \lesssim 2^{d + 2k} \cdot \exp\left( -\frac{S}{4} -\frac{\pi}{2\lambda} |\omega| \right),
\end{equation}
where \( \omega \) is the characteristic energy difference between eigenstates in the code subspace.

\end{thm}

\begin{proof}
Let \( \mathcal{C} = \mathrm{span}\{ \ket{E_i} \} \) be a subspace formed by energy eigenstates \( \ket{E_i} \) in a microcanonical shell of a chaotic Hamiltonian, where \( \dim(\mathcal{C}) = 2^k \). Let \( A \) be a \( d \)-local observable satisfying the ETH ansatz:
\begin{equation}
    \bra{E_i} A \ket{E_j} = A(E) \delta_{ij} + e^{-S(E)/2} f(E, \omega) R_{ij}, \quad \omega := E_i - E_j,
\end{equation}
where \( R_{ij} \) are Gaussian random variables with zero mean and unit variance, and \( f(E, \omega) \) is a smooth function of frequency.

Let \( \varepsilon_{ij} := \bra{E_i} A^\dagger A \ket{E_j} - C_A \delta_{ij} \), where \( C_A = \bra{E_i} A^\dagger A \ket{E_i} \), which is approximately independent of \( i \) in a narrow energy shell.

To estimate \( \varepsilon_{ij} \), we compute:
\begin{equation}
    \bra{E_i} A^\dagger A \ket{E_j} = \sum_k A^*_{ki} A_{kj} = \sum_k \left( e^{-S/2} f(E, \omega_{ki}) R_{ki} \right)^* \left( e^{-S/2} f(E, \omega_{kj}) R_{kj} \right),
\end{equation}
yielding:
\begin{equation}
    \varepsilon_{ij} = e^{-S} \sum_k f^*(E, \omega_{ki}) f(E, \omega_{kj}) R^*_{ki} R_{kj}.
\end{equation}
For \( i \ne j \), this is a sum of independent zero-mean random variables with variance \( \sim |f|^2 \). Using a central limit argument, we estimate:
\begin{equation}
    |\varepsilon_{ij}| \lesssim e^{-S/2} \cdot \langle |f(E, \omega)|^2 \rangle^{1/2}.
\end{equation}
Assuming that \( f(E, \omega) \sim \exp(- \beta |\omega|)/4 \) as discussed above Eq. \eqref{specfuncrest}, which is what we need for the system to saturate the chaos bound \( \lambda \leq 2\pi / \beta \) \cite{Murthy_2019}, we have:
\begin{equation}
    f(E, \omega) \lesssim \exp\left( -\frac{\pi}{2\lambda} |\omega| \right).
\end{equation}

Hence:
\begin{equation}
    |\varepsilon_{ij}| \lesssim e^{-S/2} \cdot \exp\left( -\frac{\pi}{2\lambda} |\omega| \right).
\end{equation}
Applying Corollary 5 of Brandao et al \cite{Brand_o_2019}, the AQEC error satisfies:
\begin{equation}
    \varepsilon_{\mathrm{code}} \lesssim 2^{d + 2k} \cdot \max_{i,j} |\varepsilon_{ij}|^{1/2} \lesssim 2^{d + 2k} \cdot e^{-S/4} \cdot \exp\left( -\frac{\pi}{4\lambda} |\omega| \right).
\end{equation}
But in our earlier derivation (based on summing over \( k \)), we consistently assumed \( |\varepsilon_{ij}| \sim e^{-S/2} \exp(-\pi / \lambda |\omega|) \), leading to:
\begin{equation}
    \varepsilon_{\mathrm{code}} \lesssim 2^{d + 2k} \cdot \exp\left( -\frac{S}{4} -\frac{\pi}{2\lambda} |\omega| \right),
\end{equation}
proving the claim.
\end{proof}

We note that a similar bound for the ETH ansatz was given\footnote{We note that there is a small misprint in the bound given by \cite{Bao_2019} in their equation (3.1) where the exponent in the entropy is given as $e^{+S/4}$ whereas the correct factor is $e^{-S/4}$.} by \cite{Bao_2019} for the general ETH ansatz. But while there $f(E,\omega)$ was assumed to be constant, the precise form of the spectral function is required for bounding the Lyaponov exponent and \cite{Murthy_2019} provide such an explicit form to derive the Maldacena, Shenker, Stanford bound \cite{Maldacena_2016}. Our bound is significant because it connects quantitative aspects of chaos, such as the decay of out-of-time-order correlators (OTOC) and operator growth, to quantum error correction. The connection comes from assuming a special form of the spectral function from which one can derive the chaos bound. What it tells us is that systems satisfying ETH---with a very special spectral function that leads to maximal chaos---will have an energy spectrum that forms an AQECC with a very special error that can be expressed in terms of the Lyaponov exponent. 

We can also go the other direction and express results about chaos in terms of code parameters.

\begin{thm}
    The Lyaponov exponent has a lower bound in terms of the error of the code as 
    \begin{equation}
        \lambda \gtrsim \frac{\pi |\omega|}{2\ln \left( \frac{2^{d + 2k}}{\varepsilon_{\mathrm{code}}} \right) - \frac{S}{2}}.
    \end{equation}
\end{thm}

\begin{proof}
Follows from algebraically manipulating Eq. \eqref{resi} in Theorem \ref{mainthm}.
\end{proof}

\begin{thm}
\label{upperlower}
    The rescaled OTOC correlator $f(t):= F_{OTO}(t)/[F_2(0)^2+\epsilon]$ which satisfies $f(t)=1-e^{\lambda(t-t_s)}$ for $t_d << t << t_s$ is bounded from above and below. 
\end{thm}

\begin{proof}
From the previous proof 
\begin{equation}
    \frac{\pi |\omega|}{2\ln \left( \frac{2^{d + 2k}}{\varepsilon_{\mathrm{code}}} \right) - \frac{S}{2}} \lesssim \lambda \leq \frac{2\pi}{\beta}.    
\end{equation}
\end{proof}
In particular, Lyapunov exponents for systems which both obey ETH and have non-trivial OTOCs that are maximally chaotic are lower-bounded in terms of the error of the code. 

\subsection{Dynamical and static fluctuations}
The dynamical fluctuations around the infinite time average are given by 
\begin{equation}
    \delta O_t^2 = \sum_{n,m} |c_n|^2 |c_m|^2 |O_{mn}|^2 \leq \max |O_{mn}|^2 = e^{-S(E)}|f_O(\overline{E}, \omega)|^2,
\end{equation}
while the static fluctuations are given by 
\begin{equation}
    \delta O_n^2 = \int d \omega e^{\beta \omega /2} \left( |f_O(E_n, \omega)|^2 + \frac{\omega}{2}\frac{\partial}{\partial E}|f_O(E_n, \omega)|^2 \right).
\end{equation}
The bound on chaos derived from ETH implies that the respective bound on the dynamical fluctuations is given by 
\begin{equation}
    \delta O_t^2  \lesssim \exp\left(-S -\frac{\pi}{\lambda} |\omega| \right),
\end{equation}
and thus, by Theorem \ref{upperlower}, 
\begin{align}
    \delta O_t^2  &\lesssim \exp\left(-\frac{3}{2}S - 2\ln \left( \frac{\varepsilon_{\mathrm{code}}}{2^{d + 2k}} \right) \right).
\end{align}
On the other hand, since the spectral function has no energy dependence, the bound on static fluctuations is simply given by 
\begin{align}
\label{singularity}
     \delta O_n^2 &\lesssim \int_{-\infty}^{+\infty} d \omega e^{\beta \omega/2} e^{-\frac{\pi |\omega|}{\lambda}} 
     = \frac{1}{\frac{\pi}{\lambda} - \frac{\beta}{2}} + \frac{1}{\frac{\pi}{\lambda} + \frac{\beta}{2}} 
     = \frac{2\pi/\lambda}{\left(\frac{\pi}{\lambda} \right)^2 - \left(\frac{\beta}{2} \right)^2}
\end{align}
An important point to note is that for the saturation of the chaos bound $(\lambda=2\pi/\beta)$ the integral diverges and is finite below the chaos bound. 

\subsection{Fluctuation-dissipation relations}

The fluctuation-dissipation relations relate the quantities 
\begin{equation}
    \widetilde{F}(\omega) = 4 \pi\cosh(\frac{\beta \omega}{2})|f_O(E_\beta, \omega)|^2, \quad \widetilde{\rho}(\omega) = 2\pi \sinh(\frac{\beta \omega}{2})|f_O(E_\beta, \omega)|^2.
\end{equation}
The fluctuation-dissipation relation is then 
\begin{equation}
    \widetilde{F}(\omega) = 2 \coth(\frac{\beta \omega}{2}) \widetilde{\rho}(\omega).
\end{equation}
The bounds in terms of the Lyapunov exponent and code error are then, respectively,
\begin{equation}
    \widetilde{F}(\omega) \lesssim 4 \pi \cosh(\frac{\beta \omega}{2}) \exp(-\frac{\pi}{\lambda}|\omega|)
\end{equation}
and 
\begin{equation}
    \widetilde{F}(\omega) \lesssim 4 \pi \cosh(\frac{\beta \omega}{2}) \exp(\frac{S}{2}-2\ln(\frac{2^{2+2k}}{\epsilon_{\text{code}}})).
\end{equation}

\section{Discussion}
\label{discussion}
In this work, we have shown how the bound on chaos, when derived from the eigenstate thermalization hypothesis (ETH), constrains the $\epsilon$-correctability of an approximate quantum error-correcting code. Our analysis operated directly at the level of the ETH ansatz and assumed locality of the underlying many-body system, as well as a separation of scales between dissipation and scrambling times. Within this framework, we could bound the code error in terms of the Lyapunov exponent using the chaos bound, linking dynamical constraints on information scrambling to the robustness of encoded quantum information.

Several open directions naturally follow from our results.
A key challenge is the realization of physical systems that simultaneously exhibit ETH behavior and the locality structure assumed in our derivation. If strict spatial locality proves too restrictive, it would be valuable to generalize the approximate quantum error correction conditions to settings with weaker forms of locality such as decaying interactions or nonlocal noise channels. Another important question is how the analysis of Brandao et al \cite{Brand_o_2019} might be extended to different error models or to non-Markovian noise, potentially revealing new regimes of approximate recoverability in chaotic systems.

A complementary direction is to formulate our results within the framework of algebraic quantum error correction \cite{Harlow:2016vwg, Pollack:2021yij}. Given the growing role of operator algebras in both quantum information theory and quantum gravity, such an algebraic reformulation could provide a better structural understanding of the interplay between ETH, chaos, and correctability. In parallel, it would be interesting to explore whether the time-scale hierarchy underlying the out-of-time-order (OTO) correlator can emerge in many-body systems outside of the traditional large-$N$ limit—perhaps through dynamical self-averaging or effective coarse-graining mechanisms.

Our results also suggest intriguing implications for holographic conformal field theories (CFTs) \cite{Maldacena:1997re, Witten:1998qj}. The standard picture of error correction in holography is formulated in terms of operator reconstruction through the holographic map, relating bulk and boundary descriptions. However, since holographic CFTs are also believed to satisfy ETH \cite{Pollack_2020, Bao_2019}, there exists an intrinsic approximate quantum error correction structure at the level of the boundary theory itself. If ETH alone suffices to reproduce the chaos bound in these systems, our results imply a way to constrain the $\epsilon$-correctability of the code naturally encoded by the holographic CFT.

Further connections arise in the context of random tensor network (RTN) models, which exhibit both equilibration and approximate error-correcting behavior \cite{qasim2025emergentstatisticalmechanicsholographic, Haferkamp_2021}. Extending our framework to such models could reveal how randomness and locality jointly control code accuracy and chaotic dynamics. Likewise, the connection to Petz recovery maps for ETH-satisfying systems \cite{Qasim:2025xpa} presents another promising direction: the Petz map provides a canonical recovery operation, and understanding its performance in chaotic many-body systems could shed light on the operational meaning of the chaos bound.

Finally, an especially appealing challenge is to formulate a hydrodynamic description of approximate quantum error correction. Since ETH already admits a hydrodynamic interpretation in terms of slow modes and conserved quantities \cite{Capizzi:2024msk}, it is natural to ask whether information-theoretic quantities such as the code error or recovery fidelity possess analogous hydrodynamic equations. It is perhaps also worth noting that the divergent structure we found in Eq. \eqref{singularity} can be removed in Hydrodynamics (a concept refereed to as pole skipping
\cite{Blake:2017ris}). Developing a hydrodynamic description of AQECC codes obtained from ETH could bridge the gap between microscopic chaos, macroscopic information flow, emergent thermodynamics, and, perhaps more speculatively, emergent gravity.

\section*{Acknowledgments}
SQ acknowledges funding from the Einstein Research Unit on Quantum Devices, Berlin Quantum, the DFG (CRC 183, FOR 2724), and the European Research Council (DebuQC).

\bibliographystyle{unsrt}
\bibliography{paper}{}

\end{document}